\documentclass[11pt]{article}
\usepackage{amsmath}
\usepackage{amsthm}
\usepackage{amssymb}
\usepackage{algorithm}
\usepackage{subfig}
\usepackage{color}
\usepackage[english]{babel}
\usepackage{graphicx}
\usepackage{grffile}
\usepackage{wrapfig,epsfig}
\usepackage{epstopdf}
\usepackage{url}
\usepackage{color}
\usepackage{epstopdf}
\usepackage{algpseudocode}
\usepackage[T1]{fontenc}
\usepackage{bbm}
\usepackage{comment}
\usepackage{dsfont}
\usepackage{physics}
\usepackage{enumitem}
\usepackage{orcidlink}

\usepackage{tikz}
\usepackage{hyperref}  %%% arxiv don't allow this.
\hypersetup{colorlinks=true,citecolor=blue,linkcolor=blue} %%% Zhao : maybe we should comment this in submission.
\usetikzlibrary{arrows}
\usepackage[margin=1in]{geometry}
\graphicspath{{./figs/}}

\usepackage[nameinlink,capitalize]{cleveref}
\usepackage{textcomp}
\usepackage{multirow}

\newtheorem{theorem}{Theorem}[section]
\newtheorem{lemma}[theorem]{Lemma}
\newtheorem{definition}[theorem]{Definition}

\newtheorem{corollary}[theorem]{Corollary}
\newtheorem{conjecture}[theorem]{Conjecture}

\newtheorem{observation}[theorem]{Observation}

\newtheorem{problem}[theorem]{Problem}
\newtheorem{open}[theorem]{Open Problem}

\newcommand{\wt}{\widetilde}

\renewcommand{\varepsilon}{\epsilon}
\renewcommand{\tilde}{\wt}

\DeclareMathOperator{\poly}{poly}

\definecolor{mygreen}{RGB}{80,180,0}
\definecolor{b2}{RGB}{51,153,255}

\newcommand{\YZ}[1]{{\color{blue}[Yuxuan: #1]}}

\title{On verifiable quantum advantage with peaked circuit sampling}
\usepackage{authblk}

\author[1]{Scott Aaronson\thanks{aaronson@cs.utexas.edu}}
\author[2,3]{Yuxuan Zhang\orcidlink{0000-0001-5477-8924}\thanks{quantum.zhang@utoronto.ca}} 

\affil[1]{Department of Computer Science, The University of Texas at Austin.}

\affil[2]{Department of Physics and Centre for Quantum Information and Quantum Control, University of Toronto}
\affil[3]{Vector Institute for Artificial Intelligence, W1140-108 College Street, Schwartz Reisman Innovation Campus, Toronto, Ontario M5G 0C6, Canada}
\date{\today}

\begin{document}

\begin{titlepage}
  \maketitle

\begin{abstract}
    Over a decade after its proposal, the idea of using quantum computers to sample hard distributions has remained a key path to demonstrating quantum advantage. \ Yet a severe drawback remains: verification seems to require classical computation that is exponential in the system size, $n$. \ As an attempt to overcome this difficulty, we propose a new candidate for quantum advantage experiments with otherwise-random ``peaked circuits'', i.e., quantum circuits whose outputs have high concentrations on a computational basis state. \ Naturally, the heavy output string can be used for classical verification. 
    
    In this work, we analytically and numerically study an explicit model of peaked circuits, in which $\tau_r$ layers of uniformly random gates are augmented by $\tau_p$ layers of gates that are optimized to maximize peakedness. \ We show that getting $1/\text{poly}(n)$ peakedness from such circuits requires $\tau_{p} = \Omega((\tau_r/n)^{0.19})$ with overwhelming probability. \ However, we also give numerical evidence that nontrivial peakedness is possible in this model---decaying exponentially with the number of qubits, but more than can be explained by any approximation where the output of a random quantum circuit is treated as a Haar-random state. \ This suggests that these peaked circuits have the potential for future verifiable quantum advantage experiments.
    
    Our work raises numerous open questions about random peaked circuits, including how to generate them efficiently, and whether they can be distinguished from fully random circuits in classical polynomial time.
\end{abstract}

\maketitle

  \thispagestyle{empty}
\end{titlepage}

\section{Introduction}
Demonstrating quantum advantage~\cite{aaronson2011computational,preskill2012quantum,aaronson2013bosonsampling,childs2013universal,aaronson2016complexity,bouland2018quantum,movassagh2018efficient}, that is, having a quantum computer perform a task that even the best classical computers would find practically impossible, has been one of the most exciting challenges in quantum computation. In 2019, Google's Sycamore processor claimed quantum advantage by performing a random circuit sampling task in 200 seconds that, they then estimated, would take the world's most powerful classical computer approximately 10,000 years to complete~\cite{arute2019quantum}. Since then, a continuous intellectual tug-of-war has taken place between various quantum advantage claims by groups like USTC and Xanadu~\cite{zhong2020quantum,madsen2022quantum,zhu2022quantum} and improved classical simulation techniques to spoof the experiments' results~\cite{huang2020classical,barak2020spoofing,gao2021limitations,pan2022solving,oh2023spoofing}. These attacks generally exploit the noise and lack of error-correction on near-term devices ~\cite{noh2020efficient,gao2021limitations,dalzell2021random}. To some skeptics, the attacks raise serious doubts about whether quantum supremacy was achieved at all—or if it was, then whether it remains achieved.

Shouldn't it be easy to defeat the classical spoofing attacks by simply scaling the experiments up? This brings us to the most fundamental drawback of the current experiments. Namely, even checking their results takes exponential classical computation.  Worse, this seems to be inherent—ironically, by the very same theoretical evidence that tells us that spoofing should be exponentially hard! Thus, if we insist on directly verifying the quantum computer’s outputs, then in practice, we can’t currently scale much beyond $\sim60$ qubits or photons. It of course helps that verification, unlike spoofing, can be done at leisure—one can spend weeks on it, burning hundreds of thousands of dollars of computer time!  Nevertheless, we seem locked into a ``cat-and-mouse game'': we can never put spoofing completely beyond classical reach, lest we put verification beyond reach as well.  As long as this remains true, it will be a question mark hanging over quantum advantage itself.

%Nevertheless, suppose even if we are able to reliably scale quantum computers, a more fundamental challenge for existing quantum advantage protocols still remains. 

But what about other possibilities for quantum advantage? Can we reach quantum advantage by, for example, running known quantum algorithms? At a high level, there are three desired properties of a near-term advantage experiment candidate:
\begin{itemize}
    \item Feasible on near-term devices
    \item Hard to simulate classically
    \item Easy to verify classically
\end{itemize}

Arguably, we know how to satisfy any two out of the three, but no current quantum protocol satisfies all three requirements. The problem, in a sentence, is that the quantum algorithms that deliver clear speedups don't seem to work on near-term devices, while the quantum algorithms that work on near-term devices don't seem to deliver clear speedups. 

%To begin with, an essential drawback for sampling-based quantum advantage methods is that they suffer from the difficulty of classical verification~\cite{harrow2017quantum}. That is, a classical verifier has difficulty knowing whether a quantum computer faithfully performed a random circuit sampling task. Next,
Reaching quantum advantage with algorithms like Shor's algorithm~\cite{shor1994algorithms} and the recent breakthrough by Yamakawa and Zhandry~\cite{yamakawa2022verifiable}, while providing convincing speedups and efficiently verifiable, requires delicate control of quantum coherence and would most likely not be feasible for near-term devices. In a different direction, despite numerous recent efforts devoted to near-term algorithms like the quantum approximate optimization algorithm (QAOA)~\cite{farhi2014quantum,farhi2016quantum,zhou2020quantum,zhang2021qed} and variational quantum eigensolver (VQE)~\cite{peruzzo2014variational,kandala2017hardware,grimsley2019adaptive,brakerski2021cryptographic,kahanamoku2022classically}, they still haven't yielded plausible signs for an in-principle quantum speed up. 
\begin{figure}[tb]
    \centering.
    \includegraphics[width=0.5\textwidth]{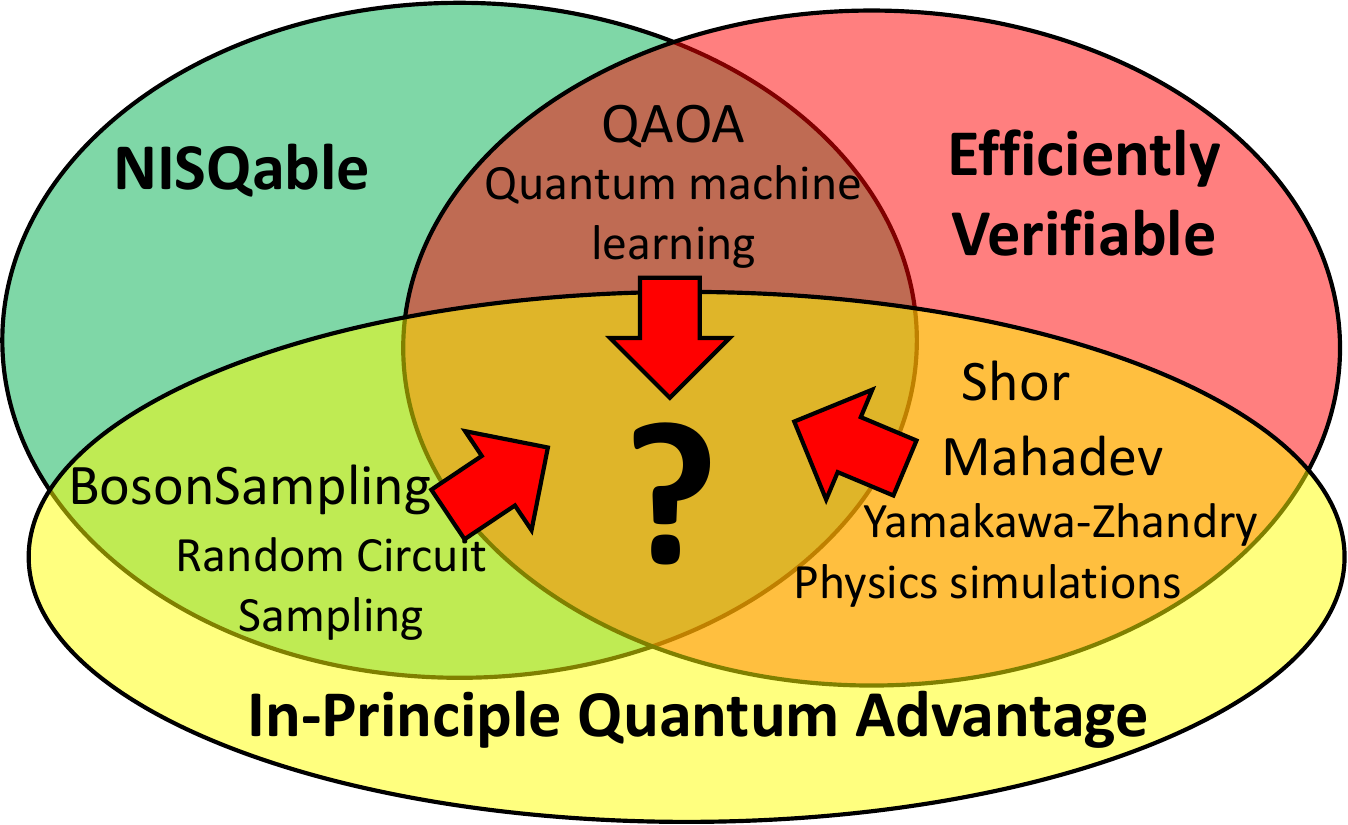}
    \caption{{\bf Field guide to verifiable quantum advantage.} Out of three possible paths to verifiable quantum advantage, we study whether it is possible to make sampling-based protocols efficiently checkable.}
    \label{fig:guide}
\end{figure}
 In a third direction, there are protocols for verifying quantumness that use, for example, post-quantum secure trapdoor claw-free functions~\cite{brakerski2021cryptographic,kahanamoku2022classically}. Still, these protocols all require delicate controls over superpositions and are likely to be unsuitable for noisy intermediate-scale quantum (NISQ)~\cite{preskill2018quantum} devices. Other NISQy proposals use pseudorandomness: A challenger creates a pseudorandom quantum circuit $C$ that conceals a secret $s$ and sends $C$ to a quantum computer; the latter has to find $s$ by running $C$. One such attempt uses exotic classical error correction codes ~\cite{shepherd2009temporally,bremner2017achieving}, but the classical security of such protocols is still under debate~\cite{kahanamoku2019forging,codsi2022classically,bremner2023iqp,bremner2023iqp,bartusek2023obfuscation,bluvstein2024logical,maslov2024fast,rajakumar2024polynomial}.
 
 In this work, we start from random circuit sampling and look for a way to make it verifiable using so-called ``peaked circuits.'' Roughly speaking, starting with the all-zero state, a peaked circuit lets the state wander around the Hilbert space, \textit{seemingly} at random, but then ending at a state that has a significant overlap with a computational basis state. 
 
 \begin{definition}[Peaked Circuit]\label{def:pc}

Given $\delta\in(0,1]$, we call the unitary $C$ $\delta$-peaked if:
$$ \max_{s\in \{0,1\}^{n}} \left|\braket{s|C}{0^n}\right|^2\geq\delta$$
with a corresponding peak weight $\delta_s \equiv \left|\braket{s|C}{0^n}\right|^2$.
\end{definition}

In other words, a circuit $C$ is $\delta$-peaked if, when applied to the all-0 initial state and then measured in the computational basis, it yields some particular output string with probability at least $\delta$.

Note that, without loss of generality, one can often focus on some particular $\delta_s$, such as $\delta_{0^n}$.  This is because, given a circuit with a peak at any $s\in \{0,1\}^n$, one can move the peak to $0^n$ by (for example) adding a small number of NOT gates at the last layer.

Peaked circuits can be used for efficient verification of quantum advantage, as follows: an challenger provides an alleged quantum computer with a circuit $C$, which is promised \textit{either} to have been chosen uniformly at random, or to have been chosen from a distribution of peaked circuits. Using a quantum computer, one can simply run $C$ a few times to see whether it is peaked (and the challenger knows which). Of course, the question remains of whether distinguishing peaked from uniformly random quantum circuits (RQCs) is hard for a classical computer.

One extremely interesting method to obtain peaked quantum circuits is to start with truly random circuits, then \textit{postselect on the event that the circuits happen to be $\delta$-peaked}, for some fixed $\delta$ (say, $1/10$).  One could then ask the question:

\begin{problem}
Can the resulting peaked circuits be distinguished, in classical polynomial time, from truly random quantum circuits?
\label{theprob}
\end{problem}

(Of course the circuits can be distinguished in \textit{quantum} polynomial time: just run them $O(1)$ times, and check whether a peak emerges in the output distribution!)

There are conflicting intuitions about the answer to Problem \ref{theprob}: on the one hand, maybe the distribution over peaked circuits is dominated by those that are peaked for ``trivial'' reasons, such as containing a large number of gates that are promptly cancelled by their inverses, giving an implementation of the identity.  On the other hand, maybe the distribution is dominated by circuits that ``take an otherwise random tour through Hilbert space,'' which \textit{so happens} (because of the postselection) to end up near a computational basis state.

Note that we could pose the same question in Nielsen's geometric picture of quantum circuits~\cite{nielsen2005geometric}.  There, the question becomes: suppose we pick a random polynomial-length path $p$ in the complexity geometry, which starts at the computational basis state $\ket{0}$, and which is constrained to end near a random computational basis state $\ket{s}$.  If we feed a description of $p$ to a classical computer, can the classical computer easily distinguish $p$ from a random unconstrained path, for example by ``contracting'' it to an essentially trivial path from $\ket{0}$ to $\ket{s}$?  Or, relatedly, can the classical computer efficiently learn the endpoint $\ket{s}$? 
 
\subsection{Our Results}

In this paper, we focus on a different but related method for generating random peaked circuits.
Consider circuits that have $\tau_{r}$ layers of random gates, followed by $\tau_{p}$ `peaking' layers of gates that can be chosen to optimize the peakedness of the overall circuit (see Fig.~\ref{fig:cir}). Note that, when $\tau_{p} \ge \tau_{r}$, we are guaranteed a 1-peaked circuit, since we can always just choose the peaking gates to invert the random gates. Thus, we are interested in what happens when $\tau_{p} < \tau_{r}$. For example,
\begin{problem}\label{prob:pc}
    Suppose $\tau_p = \tau_r / 2$, or $\tau_p = \sqrt{\tau_r}$: can we obtain a nontrivially peaked circuit \textit{then}, with high probability over the $\tau_r$ random layers?
\end{problem}

A central reason to modify the question in this way is that the original Problem~\ref{theprob} turns out to be prohibitive to study numerically, as truly random quantum circuits are very unlikely to be peaked because their output distributions are known to be not too concentrated on a small
number of output strings~\cite{hangleiter2018anticoncentration,haferkamp2020closing,brandao2021models,dalzell2022random}. A more quantitative statement can be found in Corollary~\ref{cor:rare},

Note that our new question could be rephrased as: does the state $\ket{\psi_C} = C\ket{0}$ output by a polynomial-size random quantum circuit $C$ have any ``exploitable structure'' at all to distinguish it from a Haar-random state---other than the obvious, namely that $C^{-1}\ket{\psi_C} = \ket{0}$?

Perhaps surprisingly, our numerical results provide strong evidence that the answer to this question is `yes.' In other words: whether or not it can be used for verifiable quantum advantage experiments, there \textit{is} structure in the states output by random quantum circuits, which lets us achieve nontrivial peakedness with a surprisingly small number of additional gates.  We leave the explanation of this structure as our central open problem.

What analytic results we did manage to prove are in Section~\ref{DESIGN}.  First, for $1d$ and all-to-all geometries, we show that the chance of finding $\Omega(2^{-0.49})$-peaked circuits in a logarithmic-depth RQCs is already exponentially small in $n$. Moreover, by using the known fact that polynomial-depth random quantum circuits give $t$-designs, we show that obtaining $1/\poly(n)$ peakedness, starting from a random circuit of depth $\tau$, requires $\Omega((\tau/n)^{0.19}$) gates for a $1d$ brick wall architecture.

Next, in Section \ref{PEAKABIL}, we present our numerical results. Here, we fix $\tau_p := c \tau_r$, for some $c \in (0,1)$, then we fix $\tau_r$ layers of random gates, and finally, we use gradient descent to choose the $\tau_p$ peaking layers to optimize the overall peakedness. Our central finding is that nontrivial peakedness is achievable even for constants $c \ll 1$: for example, with $n=12$ qubits and $\tau_r=40$ random layers of gates, by adding merely $\tau_p=10$ peaking layers, we can obtain, on average, a peakedness of $\delta \approx 0.2$.

Examining the actual peaked circuits, two trends emerge. First, we find that the output distributions are very similar to the distributions that arise from truly random circuits, \emph{except} that the probability of the single peaked string has been massively enhanced. Second, we find that the obtained peakedness 
$\delta$ is tightly concentrated about its mean; it fluctuates little from one circuit $C$ to another. These observations suggest, though of course they don't prove, that these peaked circuits might have ``universal,'' ``pseudorandom'' properties that make them challenging to distinguish from truly random circuits.

Admittedly, for fixed $c$ and $\tau_r/n$, the peakedness $\delta$ that we are able to achieve seems to decay exponentially with $n$.  Having said that, extrapolating our numerically-found relationships to larger $n$, we estimate that when (say) we have $n=50$ qubits, $\tau_r = 50$ random layers, and $\tau_p = 25$ peaking layers, a peakedness of $\delta \approx 0.0005$ should still be achievable, which would be feasible to detect in an experiment.\footnote{To compare, Google's first quantum advantage demonstration on 53 qubits took $5\times 10^6$ total samples over each random circuit instance~\cite{arute2019quantum}; in practice, this is more than enough to catch the peakedness even if merely $1\%$ overall state fidelity can be achieved.}

But there is an additional issue: the computation time needed to optimize the $\tau_p$ peaking layers.  On the one hand, it is entirely possible that our gradient descent algorithm got stuck in local optima, and that a better optimization method would reveal that much higher peakedness values $\delta$ are achievable; on the other hand, one might say, peaked pseudorandom quantum circuits are of limited use for quantum advantage experiments, if there is no efficient way to \textit{find} those circuits!  This leads to a second fundamental open question left by this paper:

\begin{open}
For given $\delta$, is there an efficient algorithm to generate quantum circuits that have peakedness $\delta$ and that are otherwise ``as random as possible'' (for example, that contain a large initial segment of uniformly random gates, as in this paper)?
\end{open}

\begin{figure}[tb]
    \centering.
    \includegraphics[width=0.8\textwidth]{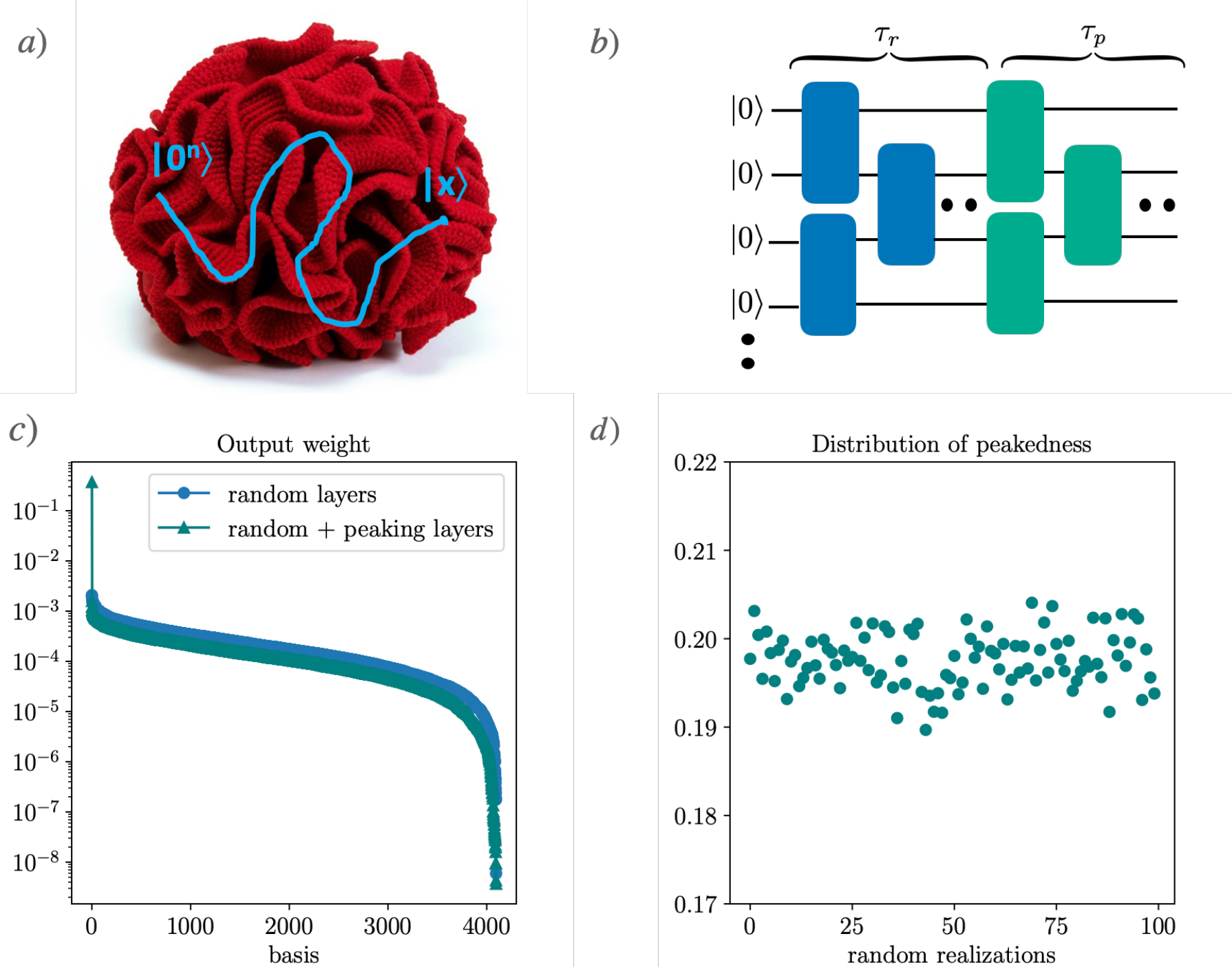}

    \caption{{\bf Generating peaked circuits from random circuits.} {\bf a)} A peaked circuit is a non-trivial short path connecting two computational bases. {\bf b)} An $1d$ circuit structure is shown here. Each block stands for a two-qubit gate: The blue blocks are drawn from the Haar-random distribution and the teal blocks are the `peaking' quantum gates. In a numerical solver, we optimize the parameters of the peaking gates, aiming at maximizing the peak weight. {\bf c)} Example of the (sorted) output weight distribution of a peaked circuit generated with our construction. At system size $n=12$, the blue curve is the output weight distribution after $\tau_r = 40$ random brickwall layers, and teal represents the distribution after merely $\tau_p = 10$ additional peaking layers. {\bf d)} we repeat the process in {\bf c} for 100 random circuits and examine the distribution over $\delta$, which turns out to have a small variance.}
    \label{fig:cir}
\end{figure}

\section{Analytical results\label{DESIGN}}

\subsection{Peaked circuits are rare in RQCs}
Given a quantum circuit $C$, let $p_C$ be the associated probability distribution over the $2^n$ possible output strings. Then $C$'s \emph{collision probability} can be defined as $\pi_C := \sum_s p_C[s]^2$. Next, given an ensemble of random circuits $C$, we will call it ``well-spread'' if the averaged collision probability is a constant times that of the maximally mixed state (which has collision probability exactly $2^{-n}$):
\begin{definition}\label{def:ac}
A circuit ensemble is well-spread if the expected collision probability satisfies
    \begin{align}\label{eq:ac}
    \mathop{\mathbb{E}}_C\big[\pi_C] \le \frac{\gamma}{2^{n}}
    \end{align}
\noindent for some constant $\gamma$.
\end{definition}

Intuitively, with overwhelming probability over the choice of circuit $C$, the output distribution cannot be too peaked. This implies the following bound:

\begin{theorem}[Probability of finding peaked circuits in an well-spread circuit ensemble]\label{theo:prob}
Let $P_\delta$ be the probability of finding a $\delta$-peaked circuit in a well-spread ensemble.  Then $P_\delta = O(\frac{1}{\delta^{2} 2^{n}})$.
\end{theorem}
\begin{proof}
Given any $\delta$-peaked circuit $C$, by Definition~\ref{def:pc}, its collision probability satisfies

$$\pi_C = \sum_s p_C[s]^2 \geq \max_{s} p_C[s]^2\geq \delta^2.$$

So taking the expectation over all $C$, we have

$$\mathop{\mathbb{E}}_C\big[\pi_C]\geq \delta^2 P_\delta.$$

Combined with Definition~\ref{def:ac} we have
$$\delta^2 P_\delta  \leq \mathop{\mathbb{E}}_C\big[\pi_C] \le \frac{\gamma}{2^{n}}.$$
\end{proof}

This shows that, in a well-spread ensemble, the probability of finding an $\delta$-peaked circuit where $\delta = \Omega(2^{-0.49n})$ is exponentially small in $n$. Furthermore, random quantum circuits for $1d$ and fully connected architectures are known to well-spread even in logarithmic depth~\cite{dalzell2022random}, which gives the following corollary:
\

\begin{corollary}[Peaked circuits are rare in logarithmic depth]\label{cor:rare}
    Assuming a $1d$, fully-connected random circuit architecture with circuit depth $\tau = \Omega(\log n)$, the probability that a random circuit $C$ is $\delta$-peaked, where $\delta = \Omega(2^{-0.49n})$, is $\exp(-\Omega(n))$.
\end{corollary}

Numerically, for example, we found that in the $1d$ brick wall architecture, with $\tau_{r} = n =10$, the maximum peakedness in 10,000 random trials was already below 0.04 (see Fig.~\ref{fig:cir}).

\subsection{How many layers are necessary to `peak' a RQC?}
For the rest of the paper, we focus on the circuit model discussed in Problem~\ref{prob:pc}. Specifically, we are interested in the relationship among $\tau_p$, $\tau_r$, and $\delta$: for example, for given $\delta$, how many peaking layers must we add to a given number of random layers to obtain a $\delta$-peaked circuit?

In this section, we prove that when $\delta = 1/\poly(n)$, the number $\tau_p$ of peaking layers needs to grow at least polynomially with $\tau_r$.

\begin{theorem}[Circuit lower bound for $1/\poly(n)$-peaked circuits]\label{theo:lb}
In a $n$-qubit system, obtaining a $\delta$-peaked circuit, where $\delta$ is at least $1/\poly(n)$, from a random quantum circuit $C$ with depth $\tau \gg n$, with high probability over $C$, requires at least $\Omega((\tau_r/n)^{0.19})$ layers.
\end{theorem}
Our lower bound can be proved with the two following results:
\begin{enumerate}
\item Polynomial-depth random circuits form approximate unitary $t$-designs~\cite{ambainis2007quantum,gross2007evenly,dankert2009exact,brandao2016local,Haferkamp2022randomquantum,mittal2023local,li2023designs}.
\item Circuits sampled from an approximate unitary $t$-design have high ``strong complexity'' (a notion to be defined later) and thus cannot be approximated with a shallow circuit~\cite{brandao2021models}.
\end{enumerate}
 \subsubsection{Unitary $t$-design} 
 A probability distribution $\mathcal{E}$ over unitary operations forms a unitary $t$-design if there is \emph{no} superoperator acting on a $t$-fold Hilbert space to distinguish it from the Haar distribution up to some given precision: 
\begin{definition}[Approximate unitary $t$-design]
    For some constant $\epsilon$, a probability distribution $\mathcal{E}$ on $U(d)$, where $d = 2^n$, forms an $\epsilon$-approximate unitary $t$-design if it obeys
    \begin{align}
        ||\Phi^{(t)}_{\mathcal{E}} -\Phi^{(t)}_{\rm Haar}||_{\Diamond}\leq\epsilon
    \end{align}
    where $\Phi^{(t)}_\mathcal{E}(A)$ is the $t$-fold moment superoperator of a operator $A$ with respect to the probability distribution $\mathcal{E}$, on $\mathcal{H}^{\otimes t}$, defined as:
    \begin{align}
        \Phi^{(t)}_\mathcal{E}(A) = \int_{\mathcal{E}} U^{\otimes t}A U^{\dag\otimes t} d\mathcal{E}(U)
    \end{align}
    
\end{definition}
Currently, the best-known result on approximate $t$-designs from random quantum circuits is provided by~\cite{Haferkamp2022randomquantum}:
\begin{lemma}[Random circuits form approximate designs]\label{lem:rc}
Random quantum circuits generate $\epsilon$-approximate unitary
$t$-designs in depth $\tau > nt^{5+o(1)}$.
\end{lemma}
\noindent The samples drawn from a unitary $t$-design are known to have well-spread properties with overwhelmingly high probability; that is, the output distribution will have a peak weight that is exponentially small in the system size $n$. 

\subsubsection{Strong circuit complexity} While circuit complexity refers to the minimum number of elementary gates needed to perform a certain computation, the ``\emph{strong} circuit complexity,'' defined by Brand$\rm\tilde{a}$o et al.~\cite{brandao2021models}, captures the difficulty of \emph{distinguishing}
a given circuit from the completely depolarizing channel, $\mathcal{D}$:
\begin{definition}[Strong circuit complexity, ancilla-free]\label{def:sc_circuit}
The strong complexity of a quantum circuit C is the minimal circuit size required to implement a measurement, with no ancilla qubits, that distinguishes $\rho \rightarrow C\rho C^{\dag}$ from the completely depolarizing channel $\mathcal{D}: \rho \rightarrow \mathbb{I}/d$ with fixed constant bias (say, $1/2$).
\end{definition}
Similarly, the strong \emph{state} complexity is defined as the minimum number of gates required to implement some measurement $M$ to \emph{distinguish} a given quantum state $\ket{\psi}$ with the maximally mixed state $\rho_0 = \mathbb{I}/d$ to some certain resolution $\eta$. Formally, let $\beta(r,\ket{\psi})$ be the maximum bias with which $\ket{\psi}$ can be distinguished from the maximally mixed state, via a circuit with at most $r$ gates from the gate set $ \mathcal{G} \subseteq U(4)$:
\begin{align}
    \beta(r,\ket{\psi}) =& \text{ max}\ |\Tr{M(\ketbra{\psi}{\psi}-\rho_0)}|\\
    &\text {s.t. }M \text{ can be implemented with at most } r \text{ gates}
\end{align}

\noindent Then the strong state complexity is defined as:
\begin{definition}[Strong state complexity]\label{def:sc_state}
    For a given $r \in \mathbb{N}$ and $\eta \in (0, 1)$, a pure state $\ket{\psi}$ has strong $\eta$-state
complexity at most $r$ if and only if
$\beta(r,\ket{\psi}) \geq 1 - 1/d - \eta$.
We denote this $\mathcal{C}_\eta (\ket{\psi}) \leq r$.
\end{definition}

It is not hard to see that the strong state complexity is an upper bound on the conventional circuit complexity, of approximately \emph{preparing} a state:
\begin{lemma}[Strong state complexity to regular state complexity]\label{lem:ssc_to_sc}
    If a quantum state $\ket{\psi}$ has strong complexity $\mathcal{C}_\eta (\ket{\psi}) \geq r$ for some $r \in \mathbb{N}$ and $\eta \in (0, 1)$, then 
    \begin{align}
        \min_{size[V]<r}\frac{1}{2}||\ketbra{\psi}{\psi} - V\ketbra{0}{0}V^\dag||_1 > \sqrt{\eta}
    \end{align}
    or equivalently, 
    \begin{align}
        \max_{size[V]<r}\frac{1}{2}\abs{\bra{0} V\ket{\psi}}^2< 1-\eta
    \end{align}
\end{lemma}

The proof of Lemma~\ref{lem:ssc_to_sc} is simply that one can invert the circuit to prepare $\ket{\psi}$.

Much more interestingly,~\cite{brandao2021models} proved a lower bound on the strong complexity of states that are generated via approximate unitary $t$-designs acting on arbitrary initial states:

\begin{theorem}[Strong state complexity of states sampled from $t$-designs]\label{theo:sc}
Consider a (pure) state in $d = 2^n$ dimensions that results from applying a randomly sampled unitary associated with an $\epsilon$-approximate 2t-design to a fixed, arbitrary initial state $\ket{\psi_0}$. Then probability that $U\ket{\psi_0}$ has strong complexity at least $r$ is:
\begin{align}
    {\rm Pr}[\mathcal{C}_\eta (\ket{\psi}) \leq r]\leq 2(1+\epsilon)d n^r|G|^r\left(\frac{16t^2}{d(1-\eta)^2}\right)^t
\end{align}
In other words, so long as $$ r \lessapprox \frac{t[n+2\log(1-\eta)-2\log (t)]}{\log(n) } $$ this probability remains tiny, provided that $n \geq |G|$ and $t < d/2$. 
 \end{theorem}
%\begin{theorem}[Strong state complexity of states sampled from $t$-designs]\label{theo:sc}
%Consider the set of (pure) states in $d = 2^n$ dimensions that results from applying all unitaries associated with an $\epsilon$-approximate 2t-design to a fixed, arbitrary initial state $\ket{\psi_0}$. Such a set must contain at least
%\begin{align}
%    {d+t-1\choose t}\left[ \frac{1}{1+\epsilon}-2d(n+1)^{r}|\mathcal{G}|^r\left(\frac{16t^2}{d(1-\eta)^2}\right)^t\right] 
%\end{align}
%distinct states with $\mathcal{C}_\eta \geq r+1.$
%In other words, so long as $$ r \lessapprox \frac{t[n-2\log (t)]}{\log(n)} $$ there are at least $(d/t)^t$ distinct such states. 
%\end{theorem}

Qualitatively, choosing any $\eta$ such that $1-\eta = 1/\operatorname{poly}(n)$, and taking the $n\rightarrow\inf$ limit where $\log(t)\ll n$, the probability of getting a circuit with strong complexity $o(nt)$ is exponentially small in $t$. Combining Lemma \ref{lem:rc} and Theorem \ref{theo:sc} gives rise to Theorem \ref{theo:lb}, which can be proved with the following argument: 
\begin{proof}[Proof of Theorem \ref{theo:lb}]
Let a unitary $C_r$ drawn from an approximate $t$-design act on the all-0 state without loss of generality.  Call the resulting state $\ket{\psi_r}$. Assume that with a high (say, constant) probability over $C_r$, there exists a circuit $C_p$ such that $C_p C_r$ is a $\delta$-peaked circuit for some $\delta = 1/\poly(n)$ and $\tau_p<t$. Then from Definition \ref{def:pc} it follows that $\abs{\bra{0}C_p\ket{\psi_r}}^2>\delta$. 

However, asymptotically in $t$, Theorem~\ref{theo:sc} says that $\ket{\psi_r}$ cannot have strong state complexity $o(nt)$; further, Lemma~\ref{lem:ssc_to_sc} shows that approximating such a state to a $1/\poly(n)$ precision requires $\Omega(nt)$ elementary gates, or $\Omega(t)$ in circuit depth. Lemma \ref{lem:rc} then implies that $t = \Omega((\tau_r/n)^{0.19})$. It follows that achieving a peakedness $\delta = \Omega(1/\operatorname{poly}(n))$ requires $\tau_p = \Omega((\tau_r/n)^{0.19})$. 
\end{proof}
It is further conjectured that $1d$ random circuits acting on qudits already form $t$-designs at depth $O(nt)$~\cite{hunter2019unitary}, which would imply $\tau_p = \Omega(\tau_r/n)$, in which case only an $O(n)$ multiplicative gap between $\tau_p$ and $\tau_r$ might suffice for a polynomially small peakedness $\delta$.  Indeed, we conjecture that a \emph{constant} multiplicative gap might already suffice for this.

\section{Numerical results on the peakability of RQCs\label{PEAKABIL}}
Having proved a lower bound on the number of layers $\tau_p$ needed for a peaked circuit, we now ask what peakedness \emph{can} be obtained with $\tau_p\ll \tau_r$, with high probability over the $\tau_r$ random layers.

We focus, without loss of generality, on $\delta_{0^n}$, the output weight of the all-$0$ string.  We denote by  $\overline{\delta_{0^n}}$ the average of this weight over all choices for the random layers.

Our goal in this section is to explore whether there is \emph{any} nontrivial peakedness to be obtained in the regime $\tau_p \ll \tau_r$.  For simplicity, we present our results for the case of a $1D$ brick wall circuit geometry, although we found quantitatively similar results in other architectures, such as all-to-all connected circuits.

In Appendix \ref{app:single}, we prove that, in the special case that the RQC consists of \emph{two} layers, there exists a single peaking layer we can add that increases the average peak weight from $(25/48)^{n/2}$ to $(7/8)^{n/2}$.

But what about the more interesting case where the RQC has a large depth---say, linear or polynomial in $n$?  Here, lacking an analytic result, we try to gain insight using a numerical optimizer based on machine learning and tensor-network packages \textsc{PyTorch}~\cite{paszke2019pytorch} and \textsc{quimb}~\cite{gray2018quimb}.

We stick with the circuit structure in Fig.~\ref{fig:cir}: the first $\tau_r$ layers consist of fixed gates that are sampled from a two-qubit Haar distribution.  They are followed by $\tau_p$ layers of parameterized quantum circuits (PQC) that depend on some variational parameter $\pmb{\theta}$. Call the resulting circuit $C({\pmb{\theta}})$.  We then run stochastic gradient-descent optimization~\cite{kingma2014adam} with the following target function:
\begin{align}
    \max_{{\pmb\theta}}\delta_{0^n}(C(\pmb\theta)) = \max_{{\pmb\theta}} \abs{\braket{0^n|C(\pmb\theta)}{0^n}}^2. 
\end{align}
\begin{figure}[tb]
    \centering.
    \includegraphics[width=0.9\textwidth]{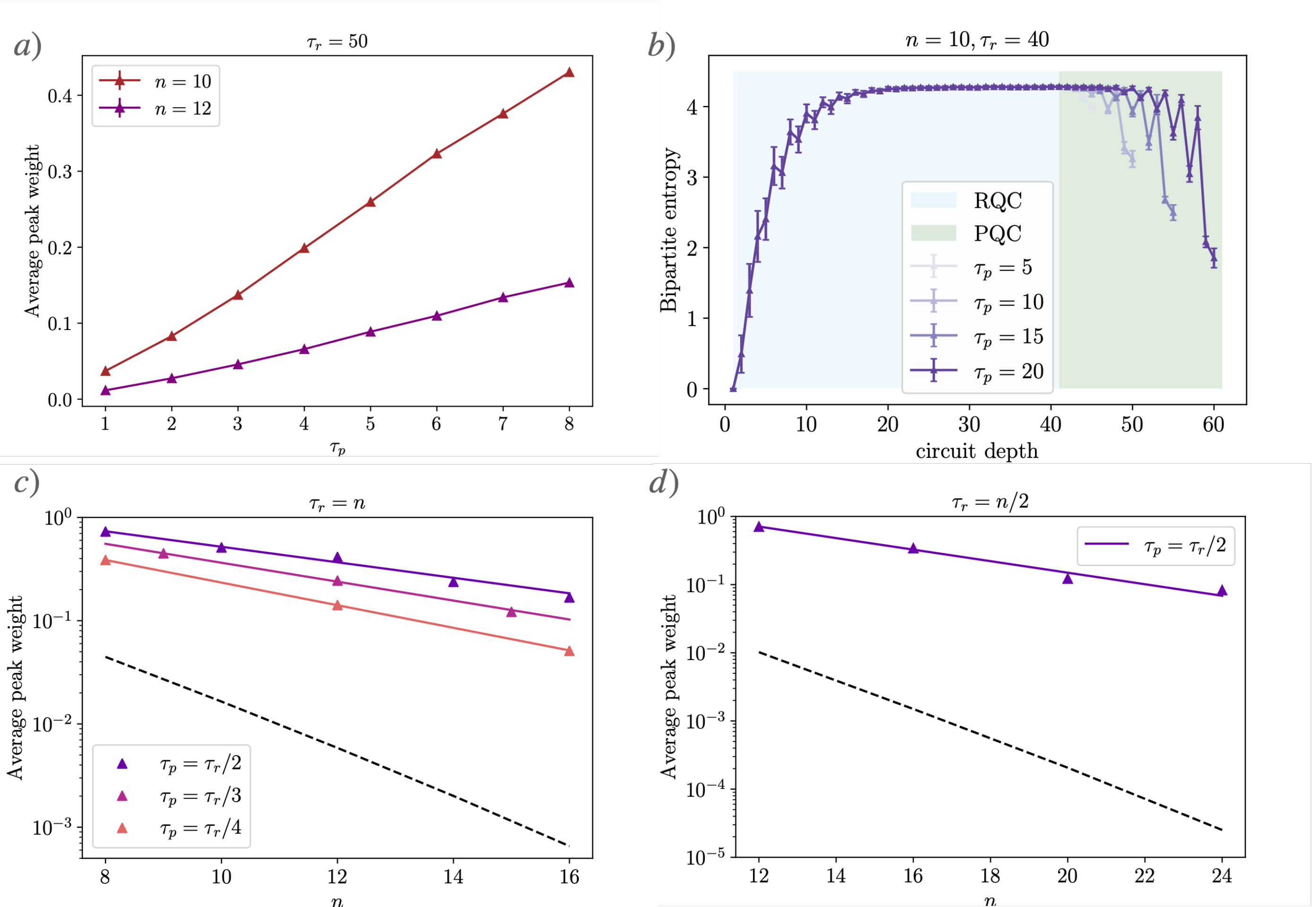}
    \caption{{\bf Numerical results on peakability.} {\bf a)} we fix the number of random layers at 50 and vary the number of peaking layers, with $n=10$ and $n=12$ qubits. Error bars for the average peak weights are too small to be visible. Each point is averaged over 100 random instances. {\bf b)} The entanglement entropy of the state at different depths, averaged over 100 instances, if we split the qubits into equal left and right halves.  {\bf c)} For each line, we fix the number of peaking layers $\tau_p$ to be a constant fraction of the number of random layers $\tau_r$, which we in turn set equal to the system size $n$. The triangles are data points and the straight lines are fitted curves. The dashed curve shows a `base' value for average peak weight at $\tau_r = n/2$. This shows that the optimized peaking layers are \emph{not} the trivial inverse of the last $\tau_p$ random layers, which could be hard for a classical distinguisher to extract the output string $s$. {\bf d)} Same as {\bf c} but now the random depth is fixed to be {$n/2$} to allow probing larger system sizes. }
    \label{fig:res}
\end{figure}
That is, we maximize the concentration of $C({\pmb{\theta}})$ on the all-zero output string. Notice that we don't need to optimize over all output strings $s$, as any two computational bases are trivially connected with at most one additional layer of NOT gates in the end, which can also be absorbed into the last layer of the PQC. 

As shown in the left panel of Figure \ref{fig:res}, we first fix $\tau_r=50$ and examine over $\tau_p$. To our surprise, attaching only $\tau_p=8$ layers led to an average peak weight of more than $0.15$ at system size $n = 10, 12$.

On the other hand, as one might expect, the reachable average peak weight drops with system size $n$ at a given ratio $\tau_p / \tau_r$. To systematically study the scalability of our peaked circuit construction, we try setting $\tau_r = n$ (bottom left) and $\tau_r = n/2$ (bottom right), and then take $\tau_p$ to be various constant fractions of $\tau_r$, namely $\tau_r/2$, $\tau_r/3$, and $\tau_r/4$.

On a log scale, it seems clear that the average peak weights are decreasing exponentially with system size $n$. Therefore we fit them onto a curve, namely
$$\text{average peak weight} = c\cdot a^{-n}.$$
When $\tau_r = n$ and $\tau_p = n/2$, for example, we get averaged peakedness that falls off as $1.189^{-n}$.  While this falls off exponentially with $n$, it falls off much less rapidly than $2^{-n}$, or even the $2^{-.49n}$ from Theorem \ref{theo:prob}. This shows that some nontrivial structure in the states output by RQCs is already being exploited.

This numerical result suggests an exponential decay in peakedness even in the regime of constant $\tau_p / \tau_r$:
\begin{conjecture}[Upper bound on average peak weight] At $\tau_r = \operatorname*{poly}(n)$ and $ \tau_p = k\tau_r$, where $0<k<1$, we have $$\text{average peak weight} = O(\operatorname*{exp}(-n)).$$
\end{conjecture}
While the exponential decay might be a little disappointing, we've found empirically that $\tau_p \ll \tau_r$ peaking layers---not nearly enough to reverse the $\tau_r$ random layers---are nevertheless enough to improve the average peakedness $\delta$ almost to the cube root of what it would've been had we merely reversed the last $\tau_p$ random layers.

\begin{conjecture}[Peaking layers improve average peak weight] At $\tau_r = \operatorname*{poly}(n)$ and $ \tau_p = k\tau_r$, where $0<k<1$, there exists a constant $0< \alpha < 1$ such that $${\text{average peak weight}(\tau_r, \tau_p,n)} = {\text{average peak weight}(\tau_r-\tau_p, 0, n)}^{\alpha}.$$
\end{conjecture}
If the pattern we observed numerically were to persist, it would imply that an average peak weight of $\sim5\times10^{-4}$ could be obtained on a noiseless 50-qubit machine, at $\tau_r = n = 50$ and $\tau_p = n/2 = 25$.  Thus, \emph{if} we also had a fast way to generate these circuits, and \emph{if} their structure was hard to detect classically, then this could provide a solution to the problem of verifiable quantum advantage on NISQ devices.

\iffalse{}
\YZ{IPR plot? entropy plot? state overlap plot?}

\subsection{Obfuscation}
We would like the peaked circuit to be immune from any local detection. As such, we provide some numerical observations here:
\begin{observation}
    we 
\end{observation}
\begin{observation}
    For a given RQC instance, multiple peaking circuits exist and they have a small overlap with each other on average.
\end{observation}
\fi{}

\section{Future directions}

A central open problem, of course, is to obtain a rigorous understanding of how much peakedness can be attained in mostly random quantum circuits.  In this paper, we showed only that obtaining $1/\poly(n)$ peakedness with high probability requires $\Omega((\tau_r/n)^{0.19})$ peaking layers, while $\tau_p = \tau_r$ layers suffice.  Our numerical results strongly indicate that nontrivial peakedness is possible with fewer peaking layers, but we do not yet understand this.  Indeed, we do not know what sort of structure in the states output by random quantum circuits could make such results possible: we leave this as a mystery.

A second open problem is how to generate these sorts of circuits efficiently. Our numerical study was limited by the infamous barren plateau issue: that is, exponentially vanishing gradient values with $n$~\cite{mcclean2018barren}. 
To increase numerical stability, we ran batches of independent optimizations with random initializations for each RQC instance and took the best result from the batch. We also attempted sequential optimization, as used in ~\cite{haghshenas2021variational,zhang2022qubit,zhang2024sequential}, but it showed no significant improvement. As we are optimizing a global function, the barren plateau issue shows up even at a constant depth and is particularly hard to mitigate~\cite{cerezo2021cost}. If we understood what sort of structure our circuits were exploiting, perhaps a completely different approach to generating such circuits would suggest itself, not reliant on gradient descent or other optimization heuristics.

A final open problem is what \emph{other} ensembles of peaked quantum circuits might be possible.  For example, what are the possibilities for obfuscating the identity transformation?  A few obfuscation protocols exist but are generally hard to implement on a NISQ device~\cite{alagic2016quantum,bartusek2023obfuscation}. Here, if we set $\tau_p \gg \tau_r$, we are guaranteed to have many different implementations of the identity; are any of them indistinguishable from random? 

\section{Acknowledgments}

We thank David Gosset, Nick Hunter-Jones, Justin Yirka, and Ruizhe Zhang for their insightful discussions. S.A. was supported by a Vannevar Bush Fellowship from
the US Department of Defense, the Berkeley NSF-QLCI CIQC Center, and a Simons Investigator Award. Y.Z. was supported by a CQIQC fellowship at the University of Toronto. This research was supported in part by grant NSF PHY-2309135 to the Kavli Institute for Theoretical Physics (KITP) and based upon Award 510817 supported by the U.S. Department of Energy, Office of Science, National Quantum Information Science Research Centers, Quantum Systems Accelerator.

The simulation code in this work can be found at \url{https://github.com/yuxuanzhang1995/Peaked-circuits}.

\newpage

\newpage
\bibliographystyle{alpha}
\bibliography{ref}

\newcommand{\etalchar}[1]{$^{#1}$}
\begin{thebibliography}{KMCVY22}

\bibitem[AA11]{aaronson2011computational}
Scott Aaronson and Alex Arkhipov.
\newblock The computational complexity of linear optics.
\newblock In {\em Proceedings of the forty-third annual ACM symposium on Theory
  of computing}, pages 333--342, 2011.

\bibitem[AA13]{aaronson2013bosonsampling}
Scott Aaronson and Alex Arkhipov.
\newblock Bosonsampling is far from uniform.
\newblock {\em arXiv preprint arXiv:1309.7460}, 2013.

\bibitem[AAB{\etalchar{+}}19]{arute2019quantum}
Frank Arute, Kunal Arya, Ryan Babbush, Dave Bacon, Joseph Bardin, Rami Barends,
  Rupak Biswas, Sergio Boixo, Fernando Brandao, David Buell, Brian Burkett,
  Yu~Chen, Zijun Chen, Ben Chiaro, Roberto Collins, William Courtney, Andrew
  Dunsworth, Edward Farhi, Brooks Foxen, and John Martinis.
\newblock Quantum supremacy using a programmable superconducting processor.
\newblock {\em Nature}, 574:505--510, 10 2019.

\bibitem[AC16]{aaronson2016complexity}
Scott Aaronson and Lijie Chen.
\newblock Complexity-theoretic foundations of quantum supremacy experiments.
\newblock {\em arXiv preprint arXiv:1612.05903}, 2016.

\bibitem[AE07]{ambainis2007quantum}
Andris Ambainis and Joseph Emerson.
\newblock Quantum t-designs: t-wise independence in the quantum world.
\newblock In {\em Twenty-Second Annual IEEE Conference on Computational
  Complexity (CCC'07)}, pages 129--140. IEEE, 2007.

\bibitem[AF16]{alagic2016quantum}
Gorjan Alagic and Bill Fefferman.
\newblock On quantum obfuscation.
\newblock {\em arXiv preprint arXiv:1602.01771}, 2016.

\bibitem[BCG20]{barak2020spoofing}
Boaz Barak, Chi-Ning Chou, and Xun Gao.
\newblock Spoofing linear cross-entropy benchmarking in shallow quantum
  circuits.
\newblock {\em arXiv preprint arXiv:2005.02421}, 2020.

\bibitem[BCHJ{\etalchar{+}}21]{brandao2021models}
Fernando~GSL Brand{\~a}o, Wissam Chemissany, Nicholas Hunter-Jones, Richard
  Kueng, and John Preskill.
\newblock Models of quantum complexity growth.
\newblock {\em PRX Quantum}, 2(3):030316, 2021.

\bibitem[BCJ23]{bremner2023iqp}
Michael~J Bremner, Bin Cheng, and Zhengfeng Ji.
\newblock Iqp sampling and verifiable quantum advantage: Stabilizer scheme and
  classical security.
\newblock {\em arXiv preprint arXiv:2308.07152}, 2023.

\bibitem[BCM{\etalchar{+}}21]{brakerski2021cryptographic}
Zvika Brakerski, Paul Christiano, Urmila Mahadev, Umesh Vazirani, and Thomas
  Vidick.
\newblock A cryptographic test of quantumness and certifiable randomness from a
  single quantum device.
\newblock {\em Journal of the ACM (JACM)}, 68(5):1--47, 2021.

\bibitem[BEG{\etalchar{+}}24]{bluvstein2024logical}
Dolev Bluvstein, Simon~J Evered, Alexandra~A Geim, Sophie~H Li, Hengyun Zhou,
  Tom Manovitz, Sepehr Ebadi, Madelyn Cain, Marcin Kalinowski, Dominik
  Hangleiter, et~al.
\newblock Logical quantum processor based on reconfigurable atom arrays.
\newblock {\em Nature}, 626(7997):58--65, 2024.

\bibitem[BFNV18]{bouland2018quantum}
Adam Bouland, Bill Fefferman, Chinmay Nirkhe, and Umesh Vazirani.
\newblock Quantum supremacy and the complexity of random circuit sampling.
\newblock {\em arXiv preprint arXiv:1803.04402}, 2018.

\bibitem[BHH16]{brandao2016local}
Fernando~GSL Brandao, Aram~W Harrow, and Micha{\l} Horodecki.
\newblock Local random quantum circuits are approximate polynomial-designs.
\newblock {\em Communications in Mathematical Physics}, 346:397--434, 2016.

\bibitem[BKNY23]{bartusek2023obfuscation}
James Bartusek, Fuyuki Kitagawa, Ryo Nishimaki, and Takashi Yamakawa.
\newblock Obfuscation of pseudo-deterministic quantum circuits.
\newblock In {\em Proceedings of the 55th Annual ACM Symposium on Theory of
  Computing}, pages 1567--1578, 2023.

\bibitem[BMS17]{bremner2017achieving}
Michael~J Bremner, Ashley Montanaro, and Dan~J Shepherd.
\newblock Achieving quantum supremacy with sparse and noisy commuting quantum
  computations.
\newblock {\em Quantum}, 1:8, 2017.

\bibitem[CGW13]{childs2013universal}
Andrew~M Childs, David Gosset, and Zak Webb.
\newblock Universal computation by multiparticle quantum walk.
\newblock {\em Science}, 339(6121):791--794, 2013.

\bibitem[CSV{\etalchar{+}}21]{cerezo2021cost}
Marco Cerezo, Akira Sone, Tyler Volkoff, Lukasz Cincio, and Patrick~J Coles.
\newblock Cost function dependent barren plateaus in shallow parametrized
  quantum circuits.
\newblock {\em Nature communications}, 12(1):1--12, 2021.

\bibitem[CvdW22]{codsi2022classically}
Julien Codsi and John van~de Wetering.
\newblock Classically simulating quantum supremacy iqp circuits trough a random
  graph approach.
\newblock {\em arXiv preprint arXiv:2212.08609}, 2022.

\bibitem[DCEL09]{dankert2009exact}
Christoph Dankert, Richard Cleve, Joseph Emerson, and Etera Livine.
\newblock Exact and approximate unitary 2-designs and their application to
  fidelity estimation.
\newblock {\em Physical Review A}, 80(1):012304, 2009.

\bibitem[DHJB21]{dalzell2021random}
Alexander~M Dalzell, Nicholas Hunter-Jones, and Fernando~GSL Brand{\~a}o.
\newblock Random quantum circuits transform local noise into global white
  noise.
\newblock {\em arXiv preprint arXiv:2111.14907}, 2021.

\bibitem[DHJB22]{dalzell2022random}
Alexander~M Dalzell, Nicholas Hunter-Jones, and Fernando~GSL Brandao.
\newblock Random quantum circuits anticoncentrate in log depth.
\newblock {\em PRX Quantum}, 3(1):010333, 2022.

\bibitem[FGG14]{farhi2014quantum}
Edward Farhi, Jeffrey Goldstone, and Sam Gutmann.
\newblock A quantum approximate optimization algorithm.
\newblock {\em arXiv preprint arXiv:1411.4028}, 2014.

\bibitem[FH16]{farhi2016quantum}
Edward Farhi and Aram~W Harrow.
\newblock Quantum supremacy through the quantum approximate optimization
  algorithm.
\newblock {\em arXiv preprint arXiv:1602.07674}, 2016.

\bibitem[GAE07]{gross2007evenly}
David Gross, Koenraad Audenaert, and Jens Eisert.
\newblock Evenly distributed unitaries: On the structure of unitary designs.
\newblock {\em Journal of mathematical physics}, 48(5), 2007.

\bibitem[GEBM19]{grimsley2019adaptive}
Harper~R Grimsley, Sophia~E Economou, Edwin Barnes, and Nicholas~J Mayhall.
\newblock An adaptive variational algorithm for exact molecular simulations on
  a quantum computer.
\newblock {\em Nature communications}, 10(1):3007, 2019.

\bibitem[GKC{\etalchar{+}}21]{gao2021limitations}
Xun Gao, Marcin Kalinowski, Chi-Ning Chou, Mikhail~D Lukin, Boaz Barak, and
  Soonwon Choi.
\newblock Limitations of linear cross-entropy as a measure for quantum
  advantage.
\newblock {\em arXiv preprint arXiv:2112.01657}, 2021.

\bibitem[Gra18]{gray2018quimb}
Johnnie Gray.
\newblock quimb: a python library for quantum information and many-body
  calculations.
\newblock {\em Journal of Open Source Software}, 3(29):819, 2018.

\bibitem[Haf22]{Haferkamp2022randomquantum}
Jonas Haferkamp.
\newblock Random quantum circuits are approximate unitary {$t$}-designs in
  depth {$O\left(nt^{5+o(1)}\right)$}.
\newblock {\em {Quantum}}, 6:795, September 2022.

\bibitem[HBVSE18]{hangleiter2018anticoncentration}
Dominik Hangleiter, Juan Bermejo-Vega, Martin Schwarz, and Jens Eisert.
\newblock Anticoncentration theorems for schemes showing a quantum speedup.
\newblock {\em Quantum}, 2:65, 2018.

\bibitem[HGPC21]{haghshenas2021variational}
Reza Haghshenas, Johnnie Gray, Andrew~C. Potter, and Garnet Kin-Lic Chan.
\newblock The variational power of quantum circuit tensor networks, 2021.

\bibitem[HHB{\etalchar{+}}20]{haferkamp2020closing}
Jonas Haferkamp, Dominik Hangleiter, Adam Bouland, Bill Fefferman, Jens Eisert,
  and Juani Bermejo-Vega.
\newblock Closing gaps of a quantum advantage with short-time hamiltonian
  dynamics.
\newblock {\em Physical Review Letters}, 125(25):250501, 2020.

\bibitem[HJ19]{hunter2019unitary}
Nicholas Hunter-Jones.
\newblock Unitary designs from statistical mechanics in random quantum
  circuits.
\newblock {\em arXiv preprint arXiv:1905.12053}, 2019.

\bibitem[HZN{\etalchar{+}}20]{huang2020classical}
Cupjin Huang, Fang Zhang, Michael Newman, Junjie Cai, Xun Gao, Zhengxiong Tian,
  Junyin Wu, Haihong Xu, Huanjun Yu, Bo~Yuan, et~al.
\newblock Classical simulation of quantum supremacy circuits.
\newblock {\em arXiv preprint arXiv:2005.06787}, 2020.

\bibitem[KB14]{kingma2014adam}
Diederik~P Kingma and Jimmy Ba.
\newblock Adam: A method for stochastic optimization.
\newblock {\em arXiv preprint arXiv:1412.6980}, 2014.

\bibitem[KM19]{kahanamoku2019forging}
Gregory~D Kahanamoku-Meyer.
\newblock Forging quantum data: classically defeating an iqp-based quantum
  test.
\newblock {\em arXiv preprint arXiv:1912.05547}, 2019.

\bibitem[KMCVY22]{kahanamoku2022classically}
Gregory~D Kahanamoku-Meyer, Soonwon Choi, Umesh~V Vazirani, and Norman~Y Yao.
\newblock Classically verifiable quantum advantage from a computational bell
  test.
\newblock {\em Nature Physics}, 18(8):918--924, 2022.

\bibitem[KMT{\etalchar{+}}17]{kandala2017hardware}
Abhinav Kandala, Antonio Mezzacapo, Kristan Temme, Maika Takita, Markus Brink,
  Jerry~M Chow, and Jay~M Gambetta.
\newblock Hardware-efficient variational quantum eigensolver for small
  molecules and quantum magnets.
\newblock {\em nature}, 549(7671):242--246, 2017.

\bibitem[LTBM08]{lakshminarayan2008extreme}
Arul Lakshminarayan, Steven Tomsovic, Oriol Bohigas, and Satya~N Majumdar.
\newblock Extreme statistics of complex random and quantum chaotic states.
\newblock {\em Physical review letters}, 100(4):044103, 2008.

\bibitem[LZL{\etalchar{+}}23]{li2023designs}
Zimu Li, Han Zheng, Junyu Liu, Liang Jiang, and Zi-Wen Liu.
\newblock Designs from local random quantum circuits with su (d) symmetry.
\newblock {\em arXiv preprint arXiv:2309.08155}, 2023.

\bibitem[MBS{\etalchar{+}}18]{mcclean2018barren}
Jarrod~R McClean, Sergio Boixo, Vadim~N Smelyanskiy, Ryan Babbush, and Hartmut
  Neven.
\newblock Barren plateaus in quantum neural network training landscapes.
\newblock {\em Nature communications}, 9(1):4812, 2018.

\bibitem[MBT{\etalchar{+}}24]{maslov2024fast}
Dmitri Maslov, Sergey Bravyi, Felix Tripier, Andrii Maksymov, and Joe Latone.
\newblock Fast classical simulation of harvard/quera iqp circuits.
\newblock {\em arXiv preprint arXiv:2402.03211}, 2024.

\bibitem[MHJ23]{mittal2023local}
Shivan Mittal and Nicholas Hunter-Jones.
\newblock Local random quantum circuits form approximate designs on arbitrary
  architectures.
\newblock {\em arXiv preprint arXiv:2310.19355}, 2023.

\bibitem[MLA{\etalchar{+}}22]{madsen2022quantum}
Lars~S Madsen, Fabian Laudenbach, Mohsen~Falamarzi Askarani, Fabien Rortais,
  Trevor Vincent, Jacob~FF Bulmer, Filippo~M Miatto, Leonhard Neuhaus, Lukas~G
  Helt, Matthew~J Collins, et~al.
\newblock Quantum computational advantage with a programmable photonic
  processor.
\newblock {\em Nature}, 606(7912):75--81, 2022.

\bibitem[Mov18]{movassagh2018efficient}
Ramis Movassagh.
\newblock Efficient unitary paths and quantum computational supremacy: A proof
  of average-case hardness of random circuit sampling.
\newblock {\em arXiv preprint arXiv:1810.04681}, 2018.

\bibitem[MP67]{Marcenko_1967}
V~A Marčenko and L~A Pastur.
\newblock Distribution of eigenvalues for some sets of random matrices.
\newblock {\em Mathematics of the USSR-Sbornik}, 1(4):457, apr 1967.

\bibitem[Nie05]{nielsen2005geometric}
Michael~A Nielsen.
\newblock A geometric approach to quantum circuit lower bounds.
\newblock {\em arXiv preprint quant-ph/0502070}, 2005.

\bibitem[NJF20]{noh2020efficient}
Kyungjoo Noh, Liang Jiang, and Bill Fefferman.
\newblock Efficient classical simulation of noisy random quantum circuits in
  one dimension.
\newblock {\em Quantum}, 4:318, 2020.

\bibitem[OJF23]{oh2023spoofing}
Changhun Oh, Liang Jiang, and Bill Fefferman.
\newblock Spoofing cross-entropy measure in boson sampling.
\newblock {\em Physical Review Letters}, 131(1):010401, 2023.

\bibitem[PCZ22]{pan2022solving}
Feng Pan, Keyang Chen, and Pan Zhang.
\newblock Solving the sampling problem of the sycamore quantum circuits.
\newblock {\em Physical Review Letters}, 129(9):090502, 2022.

\bibitem[PGM{\etalchar{+}}19]{paszke2019pytorch}
Adam Paszke, Sam Gross, Francisco Massa, Adam Lerer, James Bradbury, Gregory
  Chanan, Trevor Killeen, Zeming Lin, Natalia Gimelshein, Luca Antiga, et~al.
\newblock Pytorch: An imperative style, high-performance deep learning library.
\newblock {\em Advances in neural information processing systems}, 32, 2019.

\bibitem[PMS{\etalchar{+}}14]{peruzzo2014variational}
Alberto Peruzzo, Jarrod McClean, Peter Shadbolt, Man-Hong Yung, Xiao-Qi Zhou,
  Peter~J Love, Al{\'a}n Aspuru-Guzik, and Jeremy~L O’brien.
\newblock A variational eigenvalue solver on a photonic quantum processor.
\newblock {\em Nature communications}, 5(1):4213, 2014.

\bibitem[Pre12]{preskill2012quantum}
John Preskill.
\newblock Quantum computing and the entanglement frontier.
\newblock {\em arXiv preprint arXiv:1203.5813}, 2012.

\bibitem[Pre18]{preskill2018quantum}
John Preskill.
\newblock Quantum computing in the nisq era and beyond.
\newblock {\em Quantum}, 2:79, 2018.

\bibitem[RWL24]{rajakumar2024polynomial}
Joel Rajakumar, James~D Watson, and Yi-Kai Liu.
\newblock Polynomial-time classical simulation of noisy iqp circuits with
  constant depth.
\newblock {\em arXiv preprint arXiv:2403.14607}, 2024.

\bibitem[SB09]{shepherd2009temporally}
Dan Shepherd and Michael~J Bremner.
\newblock Temporally unstructured quantum computation.
\newblock {\em Proceedings of the Royal Society A: Mathematical, Physical and
  Engineering Sciences}, 465(2105):1413--1439, 2009.

\bibitem[Sho94]{shor1994algorithms}
Peter~W Shor.
\newblock Algorithms for quantum computation: discrete logarithms and
  factoring.
\newblock In {\em Proceedings 35th annual symposium on foundations of computer
  science}, pages 124--134. Ieee, 1994.

\bibitem[TW94]{tracy1994level}
Craig~A Tracy and Harold Widom.
\newblock Level-spacing distributions and the airy kernel.
\newblock {\em Communications in Mathematical Physics}, 159:151--174, 1994.

\bibitem[YZ22]{yamakawa2022verifiable}
Takashi Yamakawa and Mark Zhandry.
\newblock Verifiable quantum advantage without structure.
\newblock In {\em 2022 IEEE 63rd Annual Symposium on Foundations of Computer
  Science (FOCS)}, pages 69--74. IEEE, 2022.

\bibitem[ZCC{\etalchar{+}}22]{zhu2022quantum}
Qingling Zhu, Sirui Cao, Fusheng Chen, Ming-Cheng Chen, Xiawei Chen, Tung-Hsun
  Chung, Hui Deng, Yajie Du, Daojin Fan, Ming Gong, et~al.
\newblock Quantum computational advantage via 60-qubit 24-cycle random circuit
  sampling.
\newblock {\em Science bulletin}, 67(3):240--245, 2022.

\bibitem[ZJN{\etalchar{+}}22]{zhang2022qubit}
Yuxuan Zhang, Shahin Jahanbani, Daoheng Niu, Reza Haghshenas, and Andrew~C
  Potter.
\newblock Qubit-efficient simulation of thermal states with quantum tensor
  networks.
\newblock {\em Physical Review B}, 106(16):165126, 2022.

\bibitem[ZJR{\etalchar{+}}24]{zhang2024sequential}
Yuxuan Zhang, Shahin Jahanbani, Ameya Riswadkar, S~Shankar, and Andrew~C
  Potter.
\newblock Sequential quantum simulation of spin chains with a single circuit
  qed device.
\newblock {\em Physical Review A}, 109(2):022606, 2024.

\bibitem[{\v{Z}}ni06]{vznidarivc2006entanglement}
Marko {\v{Z}}nidari{\v{c}}.
\newblock Entanglement of random vectors.
\newblock {\em Journal of Physics A: Mathematical and Theoretical}, 40(3):F105,
  2006.

\bibitem[ZWC{\etalchar{+}}20]{zhou2020quantum}
Leo Zhou, Sheng-Tao Wang, Soonwon Choi, Hannes Pichler, and Mikhail~D Lukin.
\newblock Quantum approximate optimization algorithm: Performance, mechanism,
  and implementation on near-term devices.
\newblock {\em Physical Review X}, 10(2):021067, 2020.

\bibitem[ZWD{\etalchar{+}}20]{zhong2020quantum}
Han-Sen Zhong, Hui Wang, Yu-Hao Deng, Ming-Cheng Chen, Li-Chao Peng, Yi-Han
  Luo, Jian Qin, Dian Wu, Xing Ding, Yi~Hu, et~al.
\newblock Quantum computational advantage using photons.
\newblock {\em Science}, 370(6523):1460--1463, 2020.

\bibitem[ZZP21]{zhang2021qed}
Yuxuan Zhang, Ruizhe Zhang, and Andrew~C Potter.
\newblock Qed driven qaoa for network-flow optimization.
\newblock {\em Quantum}, 5:510, 2021.

\end{thebibliography}
\newpage
\appendix
\section{Improved average peak weight of a two-layer RQC from a single layer quantum circuit}\label{app:single}

In this appendix, we consider $\tau_r = 2$ and $\tau_p = 1$. We show that, given the two random layers, there exists a peaking circuit that on average polynomially increases the peak weight compared with a $\tau_r-\tau_p = 1$-layer RQC.

Assume without loss of generality that $n$ is even. Call the first and second random layers $R_1$ and $R_2$ respectively, and call the peaking layer $P$.  Observe that the output state, after we apply $R_1$ to the initial state $\ket{0^n}$, can be written as $$R_1\ket{0^n} = \prod_i^{n/2} (\alpha_i \ket{\psi_i^0}\ket{\psi_{i+1}^0}+\beta_i \ket{\psi_i^1}\ket{\psi_{i+1}^1}) = \prod_i^{n/2} U_{2i} U_{2i+1}(\alpha_i \ket{00}+\beta_i \ket{11})$$
Here we have written the state as a product of 2-local states in their Schmidt decomposed form (assuming $|\alpha_i|>|\beta_i|$), where $\{\ket{\psi_{i}^{0}},\ket{\psi_{i}^{1}}\}$ is some local orthonormal basis on the $i$-th qubit. We claim:
\begin{theorem}
    For any choice of $R_2R_1$, there exists a single-layer circuit $P$, such that the output state $PR_2R_1\ket{0^n}$ has polynomially higher average peak weight than $R_1\ket{0^n}$.
\end{theorem}
\begin{proof}

We choose $$P := \prod_i^{n/2}U_{2i}^\dag U_{2i+1}^\dag R_2^{-1}.$$

First, we calculate the average peak weight of $R_1\ket{0^n}$ over the Haar ensemble. We notice that the resulting state is simply a tensor product of 2-qubit Haar random states, whose distribution after measurement follows the famous Porter-Thomas distribution. The quantity that interests us is 
$$\int dU\ {\rm Max}|U_{0,i}|^2, $$
which can be calculated using random matrix theory~\cite{lakshminarayan2008extreme} to be $\frac{1}{d} \sum_i^d \frac{1}{d}$. For $d = 2^2 = 4$, this means the average peak weight for a 2-qubit Haar random state is $\frac{25}{48}$. Putting everything together, the average peak weight generated by single random layer is $R_1\ket{0^n}$ is $(\frac{25}{48})^{n/2}$.

On the other hand, the resulting state after the peaking circuit is $$PR_2R_1\ket{0^n} = \prod_i^{n/2} (\alpha_i \ket{00}+\beta_i \ket{11}).$$ To calculate the average peak weight, we just need to know the expectation of the highest Schmidt weight $\alpha_i$. For Haar random quantum states, the entanglement spectrum is known to follow the Marchenko-Pastur distribution~\cite{Marcenko_1967}; meanwhile, the distribution of the largest eigenvalue has been calculated~\cite{tracy1994level,vznidarivc2006entanglement}. In particular, for two-qubit random states, $\mathbb{E}[|\alpha_i|^2] = 7/8$. This means that the output state of $PR_2R_1\ket{0^n}$ has an average peak weight of $(\frac{7}{8})^{n/2}$.
\end{proof}
\end{document}